\documentclass[runningheads]{llncs}
\usepackage{graphicx}
\usepackage{amsmath}
\usepackage{amssymb}
\usepackage{amsfonts}

\usepackage{todonotes}
\usepackage{subcaption}

\newcommand\orbit{\textsf{Orbit}}

\newcommand{\cO}{\mathcal{O}}

\newcommand{\changes}{\textsf{Changes}}

\newcommand{\electionpredi}
{\textsc{Election-Prediction}}
\newcommand{\electionpred}
{\textsc{Reachability}}

\newcommand{\id}{\textsf{id}}
\newcommand{\In}{\textsf{In}}

\def\shortminus{%
  \setbox0=\hbox{-}%
  \vcenter{%
    \hrule width\wd0 height \the\fontdimen8\textfont3%
  }%
}

\begin{document}
\title{Local Certification of  Majority  Dynamics}
%
%
\author{Diego Maldonado\inst{1}\and
Pedro Montealegre\inst{2}\and
Martín Ríos-Wilson\inst{2} \and\\
Guillaume Theyssier\inst{3}}
\authorrunning{D. Maldonado et al.}
%
\institute{
Facultad de Ingeniería, Universidad Católica de la Santísima Concepción, Concepción, Chile\\
\email{dmaldonado@ucsc.cl}
\and
Facultad de Ingeniería y Ciencias, Universidad Adolfo Ibáñez, Santiago, Chile \\
\email{\{p.montealegre,martin.rios\}@uai.cl}
\and
Aix-Marseille Universit\'e, CNRS, I2M (UMR 7373), Marseille, France\\
\email{guillaume.theyssier@cnrs.fr}
}
\maketitle              
\begin{abstract}
In majority voting dynamics, a group of $n$ agents in a social network are asked for their preferred candidate in a future election between two possible choices. At each time step, a new poll is taken, and each agent adjusts their vote according to the majority opinion of their network neighbors. After $T$ time steps, the candidate with the majority of votes is the leading contender in the election. In general, it is very hard to predict who will be the leading candidate after a large number of time-steps.  

We study, from the perspective of local certification, the problem of predicting the leading candidate after a certain number of time-steps, which we call $\electionpredi$. We show that in graphs with sub-exponential growth $\electionpredi$ admits a proof labeling scheme of size $\mathcal{O}(\log n)$. We also find non-trivial upper bounds for graphs with a bounded degree, in which the size of the certificates are sub-linear in $n$. 

Furthermore, we explore lower bounds for the unrestricted case, showing that locally checkable proofs for $\electionpredi$ on arbitrary \(n\)-node graphs have certificates on  $\Omega(n)$ bits. Finally, we show that our upper bounds are tight even for graphs of constant growth.

\end{abstract}

\section{Introduction}

 Understanding social influence, including conformity, opinion formation, peer pressure, leadership, and other related phenomena, has long been a focus of research in sociology \cite{heider1946attitudes,kelman1958compliance}. With the advent of online social network platforms, researchers have increasingly turned to graph theory and network analysis to model social interactions \cite{asuncion2010turning,castellano2009statistical,mislove2007measurement}. In particular, opinion formation and evolution have been extensively studied in recent years \cite{bredereck2017manipulating,javarone2014network,javarone2014social,moussaid2013social,nguyen2020dynamics,yildiz2013binary}.

One of the simplest and most widely studied models for opinion formation is the majority rule \cite{castellano2009statistical}. In this model, the opinion of an individual evolves based on the  opinion of the majority of its neighbors. Specifically, consider an election with two candidates, labeled as 
\(0\) and \(1\), and let \(G\) be an undirected, connected, and finite graph representing the social network. Each node in the graph represents an individual with a preference for the candidate they will vote for. We call this preference \emph{an opinion.} A particular assignment of opinions to each node is called a \emph{configuration}. At the beginning,  we consider that an initial configuration is given, representing the personal beliefs of each individual about their vote intentions. The configuration evolves in synchronous time-steps, where each individual updates its opinion according to the opinion of the majority of its neighbors. If the majority of its neighbors plan to vote for candidate $1$, the node changes its opinion to $1$. Conversely, if the majority of its neighbors prefer $0$, the node switches to $0$. In the event of a tie, where exactly half of the neighbors favor $1$ and the other half favor $0$, the individual retains their current opinion.

All graphs have configurations in which (locally) every node has the same opinion as the majority of its neighbors. These configurations are called \emph{fixed points} since each node retains its opinion in subsequent time-steps. Interestingly,  every graph admit non-trivial fixed points, where the local majority opinion of some nodes is different from the global majority. In general, the initial global majority opinion (among all nodes) is not preserved when the opinions evolve in the majority dynamics. In fact, the majority opinion can shift between the two candidates in non-trivial ways, which depend on both the properties of the graph and the initial distribution of opinions within it.

Actually, some initial configurations never converge to a fixed point. For example, in a network consisting of only two vertices connected by an edge, where one vertex initially has opinion 0 and the other has opinion 1, the two nodes exchange their opinions at each time-step, never reaching a fixed point. Configurations with this behavior are called limit cycle of period 2. Formally, a limit cycle of period 2 is a pair of configurations that mutually evolve into one another under the majority dynamics. In \cite{goles1985decreasing}, it was shown that for any initial configuration on any graph, the majority dynamics either converge to a fixed point or a limit cycle of period 2. The results of \cite{goles1985decreasing} also imply that the described limit behavior (called \emph{attractor}) is reached after a number of time-steps that is bounded by the number of edges in the input graph.

Therefore, even assuming that the opinion a society evolves according to the majority rule, deciding \emph{who wins the election} is a non-trivial task. During the last 25  years, there has been some effort in characterizing the computational complexity of this problem \cite{concha2022complexity,goles2016pspace,moore1997majority}. In this article we tackle this problem from the perspective of distributed algorithms and local decision. \\

\noindent{\bf Local decision.} Let \(G =(V,E)\) be a simple connected \(n\)-node graph. A \emph{distributed language} \(\mathcal{L}\) is a (Turing-decidable) collection of tuples \((G,\id, \In)\), called \emph{network configurations}, where \(\In: V\rightarrow \{0,1\}^*\) is called an \emph{input function} and \(\id: V \rightarrow [n^c]\) is an injective function that assigns to each vertex a unique identifier in \(\{1, \dots, n^c\}\) with \(c>1\). In this article, all our distributed languages are independent of the \(\id\) assignments. In other words, if \((G,\id, \In) \in \mathcal{L}\) for some \(\id\), then \((G,\id', \In) \in \mathcal{L}\) for every other \(\id'\). 

Given \(t>0\), a \emph{local decision algorithm} for a distributed language \(\mathcal{L}\) is an algorithm on instance \((G, \id, \In)\), each node \(v\) of \(V(G)\) receives the subgraph induced by all nodes at distance at most \(t\) from \(v\) (including  their identifiers and inputs).  The integer \(t>0\) depends only on the algorithm (not of the size of the input). Each node  performs unbounded computation on the information received, and decides whether to accept or reject, with the following requirements: 
\begin{itemize}
\item When \((G,\id, \In) \in \mathcal{L}\) then every node accepts. 
\item When \((G,\id, \In) \notin \mathcal{L}\) there is at least one vertex that rejects.
\end{itemize}

\noindent{\bf Distributed languages for majority dynamics.} Given an graph \(G\) and an initial configuration \(x\). The \emph{orbit} of \(x\), denoted \(\orbit(x)\), is the sequence of configurations \(\{x^t\}_{t>0}\) such that \(x^0 = x\) and for every \(t>0\), \(x^{t}\) is the configuration obtained from \(x^{t-1}\) after updating the opinion of every node under the majority dynamics. We denote by \(\electionpredi\) the set of triplets \((G,x,T)\), where the majority of the nodes vote \(1\) on time-step \(T\) starting from configuration \(x\). Formally,

 \[\electionpredi = \left\{ ( G,(x,T) ) : \begin{array}{l} x \textrm{ is a configuration of } V(G),\\T>0,\\ \textrm{ and } \sum_{v\in V(G)} x_v^T > \frac{|V(G)|}{2}. \end{array}\right\}\]

It is easy to see that there are no local decision algorithms for \electionpredi. That is, there are no algorithms in which every node of a network exchange information solely with nodes in its vicinity, and outputs which candidate wins the election. Intuitively, a local algorithm solving \electionpredi\ requires the nodes to count the states of other nodes in remote locations of the input graph. In fact, this condition holds even when there is no dynamic, i.e. \(T=0\). In that sense, the counting difficulty of problem \electionpredi\ hides the complexity of predicting the majority dynamics.  For that reason, we also study the following problem:

\[\electionpred = \left\{  (G,(x,y,T))  : \begin{array}{l} x \textrm{ is a configuration of } V(G),\\T>0,\textrm{ and } x^T = y. \end{array}\right\}.\]

Problem \electionpred\ is also hard from the point of view of local decision algorithms. Indeed, the opinion of a node \(v\) after \(t\) time-steps depends on the initial opinion of all the nodes at distance at most \(t\) from \(v\). There are graphs for which the majority dynamics stabilizes in a number of time-steps proportional to the number of edges of the graph \cite{goles1985decreasing}. Therefore, no local algorithm will be able to even decide the opinion of a single vertex in the long term. \\

\noindent{\bf Local certification.} A locally checkable proof for a distributed language \(\mathcal{L}\) is a prover-verifier pair where the prover is a non-trustable oracle assigning certificates to the nodes, and the verifier is a distributed algorithm enabling the nodes to check the correctness of the certificates by a certain number of communication rounds with their neighbors. Note that the certificates may not depend on the instance \(G\) only, but also on the identifiers \(\id\) assigned to the nodes. In proof-labeling schemes, the information exchanged between the nodes during the verification phase is limited to the certificates. Instead, in locally checkable proofs, the nodes may exchange extra-information regarding their individual state (e.g., their inputs \(\In\) or their identifiers, if not included in the certificates, which might be the case for certificates of sub-logarithmic size). The prover-verifier pair must satisfy the following two properties.
 \medskip

\noindent{\bf Completeness}: Given \((G,\In)\in \mathcal{L}\), the non-trustable prover can assign certificates to the nodes such that the verifier accepts at all nodes;

\medskip

\noindent{\bf Soundness:} Given \((G,\In)\notin \mathcal{L}\), for every certificate assignment to the nodes by the non-trustable prover, the verifier rejects in at least one node.

 \medskip

The main complexity measure for both locally checkable proofs, and proof-labeling schemes is the size of the certificates assigned to the nodes by the prover. Another complexity measure is the number of communication rounds executed during the verification step. In this article, all our upper-bounds are valid for Proof Labeling Schemes with one-round certification, while all our lower-bounds apply to locally ceckable proofs with an arbitrary number of rounds of verification.

\subsection{Our results}

We show that in several families of graphs there are efficient certification protocols for \(\electionpred\). More precisely, we show that there is a proof labeling scheme for \(\electionpred\)  with certificates on \(\cO(\log n)\) bits in \(n\)-node networks of sub-exponential growth. 

A graph has sub-exponential growth if, for each node \(v\), the cardinality of the set of vertices at distance at most \(r\) from \(v\) growths as a sub-exponential function in \(r\), for every~\(r>0\). Graphs of sub-exponential growth have bounded degree, and include several structured families of graphs such as the \(d\)-dimensional grid, for every~\(d>0\). Nevertheless, not every class of graphs of bounded degree is of sub-exponential growth. For instance, a complete binary tree has exponential growth.

For graphs of bounded degree, we show that \(\electionpred\) admits proof labeling schemes with certificates of sub-linear size. More precisely, we show that there is a proof labeling scheme for \(\electionpred\) with certificates on \(\cO(\log^2 n)\) bits in \(n\)-node networks of maximum degree \(3\). Moreover, for each \(\Delta>3\) there exists a \(\epsilon > 0\) such that is a proof labeling scheme for \(\electionpred\) with certificates on \(\cO(n^{1-\epsilon})\) bits in \(n\)-node networks of maximum degree \(\Delta>3\). Then, we show that all our upper-bounds are also valid for \(\electionpredi\). 

Then, we focus on lower-bounds. First, we show that in unrestricted families of graphs every, locally-checkable proof for the problem \(\electionpred\) as well as \(\electionpredi\) requires certificates of size \(\Omega(n)\). We also show that even restricted to graphs of degree \(2\) and constant growth, every locally-checkable proof for \(\electionpred\)  requires certificates of size \(\Omega(n)\). \\            

\noindent{\bf Our techniques.} Our upper bounds for the certification of \(\electionpred\) are based on an analysis of the maximum number of time-steps on which an individual may change its opinion during the majority dynamics.  This quantity is in general unbounded. For instance, in an attractor which is a cycle of period two an oscillating node switches its state an infinite number of time-steps. However,  when we look to two consecutive iterations of the majority dynamic, we obtain that the number of changes of a given node depends on the topology of the network. We show that the in the dynamic induced by two consecutive iterations of the majority dynamic (or, alternatively, looking one every two time-steps of the majority dynamics), the number of times that a node switches is state is constant on graphs of sub-exponential growth, logarithmic on graphs of maximum degree \(3\) and sublinear on graphs of bounded degree. 
The bound for graphs of sub-exponential growth was observed  in \cite{ginosar2000majority}, while the other two bounds can be obtained by a careful analysis of the techniques used in \cite{ginosar2000majority} (see Section \ref{sec:majoritybounded} for further details). Roughly speaking, the idea consists defining a function that assigns to each configuration a real value denoted the \emph{energy} of the configuration. This energy function is strictly decreasing in the orbit of a configuration before reaching an attractor. In fact, through the definition of such function it can be shown that the majority dynamics reaches only fixed points or limit cycles of period two, in at most a polynomial number of time-steps \cite{goles1985decreasing}. In this article, we analyze a different energy function proposed in \cite{ginosar2000majority}, from which obtain the upper-bounds for the number of times that a node can switch states in two-step majority dynamic. 

Our efficient proof labeling schemes are then defined by simply giving each vertex the list of time-steps on which it switches it state. From that information the nodes can reconstruct their orbit. We show that the nodes can use the certificates of the neighbors to verify that the recovered orbit corresponds to its real orbit under the majority dynamic. 
Our upper bounds for the certification of  \(\electionpredi\) follow from the protocol used to certify \(\electionpred\), and the use of classical techniques of local certification to count the total number of nodes in the graph, as well as the number of vertices that voted for each candidate.

Our lower-bounds are obtained using two different techniques. First, we show that in unrestricted families of graphs, every locally-checkable proof for the problem \(\electionpred\) or \(\electionpredi\) requires certificates of size \(\Omega(n)\) by a reduction to the disjointedness problem in non-deterministic communication complexity. Then, we prove that restricted to graphs of degree \(2\) and constant growth, every locally-checkable proof for \(\electionpred\)  requires certificates of size \(\Omega(n)\) by using a locally-checkable proof for \(\electionpred\) to design a locally-checkable proof that accepts only if a given input graph is a cycle. In \cite{goos2016locally} it is shown that any locally checkable proof for the problem of distinguishing between a path of a cycle requires certificates of size \(\Omega(\log n)\), implying that certifying \electionpred\ on graphs of degree at most \(2\) and constant growth (a cycle or a path) also requires  \(\Omega(\log n)\) certificates.

\subsection{Related Work}

{\bf Majority dynamics for modeling  social influence} Numerous studies have been conducted on the majority dynamics. In \cite{PhysRevE.68.046106}, the authors studied how noise affects the formation of stable patterns in the majority dynamics. They found that the addition of noise can induce pattern formation in graphs that would otherwise not exhibit them. In \cite{nguyen2020dynamics} the authors explore opinion dynamics on complex social networks, finding that densely-connected networks tend to converge to a single consensus, while sparsely-connected networks can exhibit coexistence of different opinions and multiple steady states. Node degree influences the final state under different opinion evolution rules. Variations of the majority dynamics have been proposed and studied, such as the noisy majority dynamics \cite{vieira2016phase}, where agents have some probability of changing their opinion even when they are in the local majority, and the bounded confidence model \cite{deffuant2000mixing,hegselmann2002opinion}, where agents only interact with others that have similar opinions.\\

\noindent {\bf Complexity of Majority Dynamics.} Our results are in the line of a series of articles that aim to understand the computational complexity of the majority rule by studying different variants of the problem.  In that context, two perspectives have been taken in order to show the P-Completeness. In \cite{goles2015complexity} it is shown that the prediction problem for the majority rule is P-Complete, even  when the topology is restricted to planar graphs where every node has an odd number of neighbors. The result is based in a crossing gadget that use a sort of \emph{traffic lights}, that restrict the flow of information  depending on the parity of the time-step. Then in \cite{goles2014computational} it is shown that the prediction problem for the majority rule is P-Complete when the topology is restricted to regular graphs of degree $3$ (i.e. each node has exactly three neighbors).  In \cite{goles2018complexity} study the majority rule in two dimensional grids where the edges have a \emph{sign}. The \emph{signed majority} consists in a modification of the majority rule, where the most represented state in a neighborhood is computed multiplying the state of each neighbor by the corresponding sign in the edge. The authors show that when the configuration of signs is the same on every site (i.e. we have an homogeneous cellular automata) then the dynamics and complexity of the signed majority is equivalent to the standard majority.  Interestingly, when the configuration of signs may differ from site to site, the prediction problem is P-Complete. 
 A last variant considers the prediction problem under a sequential updating scheme. More precisely, the \emph{asynchronous prediction} problem asks for the existence of a permutation of the cells that produces a change in the state of a given cell, in a given time-step. In fact, in \cite{moore1997majority} Moore suggested in this case it holds a similar dichotomy than in the synchronous case: namely, the complexity in the two-dimensional case is lower than in three or more dimensions. This conjecture was proven in \cite{goles2020complexity} where it was shown that the asynchronous prediction in two dimension is in NC, while it is NP-Complete in three or more dimensions. \\

\noindent {\bf Local certification.} Since the introduction of PLSs~\cite{KormanKP10}, different variants were introduced. As we mentioned, a stronger form of PLS are locally checkable proofs~\cite{goos2016locally}, where each node can send not only its certificates, but also its state and look up to a given radious. Other stronger forms of local certifications are $t$-PLS~\cite{FeuilloleyFHPP21}, where nodes perform communication at distance~$t\geq 1$ before deciding. Authors have studied many other variants of PLSs, such as randomized PLSs \cite{fraigniaud2019randomized}, quantum PLSs~\cite{FraigniaudGNP21}, interactive protocols~\cite{CrescenziFP19,kol2018interactive,NaorPY20}, zero-knowledge distributed certification~\cite{BickKO22}, 
PLSs use global certificates in addition to the local ones~\cite{FeuilloleyH18}, etc. On the other hand, some trade-offs between the size of the certificates and the number of rounds of the verification protocol have been exhibited ~\cite{FeuilloleyFHPP21}. Also, several hierarchies of certification mechanisms have been introduced, including games between a prover and a disprover~\cite{BalliuDFO18,FeuilloleyFH21}.

PLSs have been shown to be effective for recognizing many graph classes. For example, there are compact PLSs  (i.e. with logarithmic size certificates) for the recognition of acyclic graphs \cite{KormanKP10}, planar graphs~\cite{feuilloley2020compact}, graphs with bounded genus~\cite{EsperetL22}, \(H\)-minor-free graphs (as long as \(H\) has at most four vertices)~\cite{BousquetFP21}, etc. In a recent breakthrough, Bousquet et al.~\cite{bousquet2021local}  proved a ``meta-theorem'',  stating that, there exists a PLS for deciding any monadic second-order logic property with $O(\log n)$-bit certificates on graphs of bounded \emph{tree-depth}. This result has been extended by Fraigniaud et al~\cite{FraigniaudMRT22} to the larger class of graphs with bounded \emph{tree-width}, using certificates on $O(\log^2 n)$ bits. 

Up to our knowledge, this is the first work that combines the study of majority dynamics and local certification. 

\section{Preliminaries}

Let \(G = (V,E)\) be a graph. We denote by \(N_G(v)\) the set of neighbors of \(v\),  formally \(N_G(v) = \{u \in V: \{u,v\} \in E\}\). The \emph{degree} of \(v\), denoted \(d_G(v)\) is the cardinality of \(N_G(v)\). The \emph{maximum degree} of \(G\), denoted \(\Delta_G\), is the maximum value of \(d_G(v)\) taken over all \(v\in V\).
We say that two nodes \(u,v\in V\) are \emph{connected} if there exists a path in \(G\) joining them. In the following, we only consider connected graphs.  The \emph{distance} between \(u,v\), denoted \(d_G(u,v)\) is the minimum length (number of edges) of a path connecting them. The \emph{diameter} of \(G\) is the maximum distance over every pair of vertices in \(G\). For a node \(v\), and \(k\geq 0\), the \emph{ball of radius \(k\) centered in \(v\)}, denoted by \(B(v,k)\) is the set of nodes at distance at most \(k\) from \(v\). Formally,
\[B_G(v,k) = \{u \in V: d_G(v,u) \leq k\}\]
We also denote by \(\partial B_G(v,k) = B_G(v,k) \setminus B_G(v,k-1)\) the border of \(B_G(v,k)\).
In the following, we omit the sub-indices when they are obvious by the context. 

Let  \(G = (V,E)\) be a graph, \(v\in V\) be an arbitrary node and  \(f:\mathbb{N}\rightarrow \mathbb{R}\) a function. We say that \(v\) has a \(f\)-bounded growth if there exist constants \(c_1, c_2>0\) such that, for every \(k>0\),  \(c_1 f(k) \leq |\partial B(v,k)| \leq c_2 f(k)\) . We also say that  \(G\) has \(f\)-bounded growth if every node \(v\) has \(f\)-bounded growth. A family of graphs \( \mathcal{G}\) has \(f\)-
bounded growth if every graph in \(\mathcal{G}\) has \(f\)-bounded growth. A family  of graphs \(\mathcal{G}\) has constant-growth (respectively linear, polynomial, sub-exponential, exponential)-growth if \(\mathcal{G}\) has \(f\)-bounded growth, with \(f\) a constant (resp. linear, polynomial, sub-exponential, exponential) function. Observe that since \(B(v,1) = N(v)\), for every \(f\)-bounded graph \(G\) we have that \(\Delta(G) \leq f(1)\).

\subsection{Majority and finite state dynamics.} Let $G=(V,E)$ be a connected graph. We assign to each node in $G$ an initial opinion $v \mapsto x_v \in \{0,1\}.$ We call $x$ a \emph{configuration} for the network $G.$ We call $x(t)$ the \emph{configuration} of the network in time $t.$ We define the majority dynamics in $G$ by the following local rules for $u \in V$ and $t\geq0$:
\[x^{t+1}_u = \begin{cases}
1 & \text{ if } \sum \limits_{v \in N(u)} x^t_v  > \frac{d(v)}{2}, \\
x^t_u  &\text{ if } \sum \limits_{v \in N(u)} x^t_v =  \frac{d(v)}{2}, \\
0 & \text{ otherwise. }  
\end{cases}\]
where $x^0 = x$ is called an \emph{initial configuration}. 
Notice that in the tie case (i.e. a node observe the same number of neighbors in each state), we consider that the node remains in its current state. Therefore, nodes of even degree may depend on their own state while nodes of odd degree do not. Therefore, we can also define the local rule of the majority dynamics as follows:
\[x^{t+1}_u = \textrm{sgn}\left( \sum_{v \in V} a_{uv}x_v^{t} - \dfrac{d(u)}{2}\right)\]
where 
\[a_{uv} =  \begin{cases}
    1 &\text{ if } uv \in E, \\
    1 &\text{ if } u = v \text{ and } d(u) \text{ is even,}\\
    0 &\text{ otherwise. }
\end{cases}\]
and \(\textrm{sgn}(z)\) is the function that equals \(1\) when \(z>0\) and \(0\) otherwise.

Given a configuration \(x\) of a graph \(G\) and a vertex \(v\in V(G)\), we define the \emph{orbit} of \(x\) as the sequence of states that  \(\orbit(x) = (x^0=x, x^1,x^2,\hdots)\) that the majority dynamics visit when the initial configuration is \(x\). We also define the \emph{orbit of vertex} \(v\) as the sequence \(\orbit(x) = (x^0_v=x_v, x^1_v,x^2_v,\hdots,)\).

Observe that the orbit of every configuration is finite and periodic. In other words, there exist non-negative integers \(T=T(G,x)\) \(p=p(G,x)\) such that \(x^{T+p} = x^T\). Indeed, in an \(n\)-node graph it is possible to define exactly \(2^n\) possible configurations. Therefore, in every orbit there is at least one configuration that is visited twice. The minimum \(T\) and \(p\) satisfying previous condition are denoted, respectively, the \emph{transient length} and the \emph{period} of configuration \(x\). The \emph{transient length} \(\textsf{Transient}(G)\) and the \emph{period} \(\textsf{Period}(G)\) of graph a \(G\) are defined, respectively, as the maximum transient length and period of over all configurations of \(G\). Formally,
\[\textsf{Transient}(G) = \max \{T(G,x) : x \in \{0,1\}^V\}, \] \[\textsf{Period}(G) = \max \{p(G,x) : x \in \{0,1\}^V\}.\]

Let  \(x\) a configuration satisfying that \(T(x,G) =  0\) is called an \emph{attractor}. An attractor \(x\) satisfying \(p(x,G) = 1\) is denoted a \emph{fixed point}. Otherwise, it is denote a \emph{limit-cycle} of period \(p(x,G)\).

\subsection{Limit behavior of majority dynamics }

In \cite{goles1985decreasing} it is shown that the transient length of the majority dynamics over every graph \(G = (V,E)\) is at most \(|E|\). Moreover, the all the attractors are either fixed-points or limit-cycles of period \(2\). For sake of completeness, we give a full proof of this result.

\begin{proposition}\label{prop:energy}
For every graph \(G\), \(\textsf{Transient}(G) \leq |E|\) and \(\textsf{Period}(G) \leq~2\). 
\end{proposition}

Let us fix a graph \(G\) a configuration \(x \in \{0,1\}^V\). The \emph{energy} of the orbit of \(\textsf{Orbit}(x)\) is a function that assigns to each time step \(t\) the a value \(E^t(x)\) as follows: 
\[E^t(x)=\sum_{u,v\in V}a_{uv}|x_u^{t+1}-x_v^{t}|\]

\begin{lemma}\label{lem:energy}
 \(E^{t+1}(x) < E^{t}\) for every \(t < T(G,x)\).  
\end{lemma}

\begin{proof}
Observe that, since \(G\) is undirected, \[\sum_{u,v\in V}a_{uv}|x_u^{t}-x_v^{t-1}| = \sum_{u,v\in V}a_{uv}|x_u^{t-1}-x_v^{t}|.\]
 Then,
\begin{align*}
 E^{t+1}(x)-E^t(x) & = \sum_{u,v\in V}a_{uv}|x_u^{t+1}-x_v^{t}| - \sum_{u,v\in V}a_{uv}|x_u^{t}-x_v^{t-1}|\\
&= \sum_{u,v\in V}a_{uv}\left(|x_u^{t+1}-x_v^{t}| - |x_u^{t-1}-x_v^{t}|\right)\\
&= \sum_{u \in V}\sum_{v\in N(u)}\left(|x_u^{t+1}-x_v^{t}| - |x_u^{t-1}-x_v^{t}|\right)
\end{align*}
 For each \(u \in V\) let us call \(\Sigma_u = \sum \limits_{v \in N(u)} \left(|x_u^{t+1}-x_v^{t}| - |x_u^{t-1}-x_v^{t}|\right)\).  If \(t>0\) is such that \(x_u^{t+1} = x_u^{t-1}\), then we have \(\Sigma_u = 0\). Otherwise, by the majority rule, at time $t,$ we have that the majority  of the neighbors of \(u\) are in state  \(x_u^{t+1}\) which implies that \(\Sigma_u < 0\). 
\end{proof}

The proof Proposition~\ref{prop:energy} is a direct consequence of the previous lemma.

\begin{proof}[Proof of Proposition \ref{prop:energy}]
From Lemma~\ref{lem:energy} we have that the energy  is strictly decreasing in the transient of a orbit. Observe that for every configuration \(x\) and every \(t>0\) we have that \(0 \leq E^t(x) \leq |E|\). Moreover, if \(t \leq T(G,x)\) then \(E^{t+1}(x) - E^{t}(x) \leq -1\). Therefore, in \(t \leq |E|\) time-steps the orbit satisfies that \(E^{t+1}(x) = E^t(x)\). Such a time-step must satisfy \(x^{t+1}_u = x^{t-1}_u\) for every \(u\in V\). In other words, the configuration reached is a fixed point or a limit cycle of period~\(2\). 
\end{proof}

\section{Majority on graphs of bounded degree}\label{sec:majoritybounded}

In this section, we focus in the case of networks of bounded degree. Our analysis is based on the results of \cite{ginosar2000majority}, where the authors aim to study the asymptotic behavior of the majority dynamics on infinite graphs. Our goal is to bound the number of changes in the two-step majority dynamics. We consider a variant of the energy operator. More precisely, for each \(t>0\) and \(u\in V\), we aim to bound the number of time-steps on which the quantity \(c^t_u(x) = |x^{t+1}_{u} - x^{t-1}_{u}|\) is non-zero.


\begin{theorem}\label{theo:boundedchange}
Let \(G\) be a graph and \(x\) be an arbitrary configuration. Then, for every \(r\in V\) and every \(T>0\), it holds:
\begin{enumerate}

 \item If \(G\) is a  graph of sub-exponential growth, then 
\(\displaystyle{\sum_{t=1}^T c_r^t(x)  = \mathcal{O}(1)} \).
\item If \(G\) is a graph of maximum degree \(3\), then \(\displaystyle{\sum_{t=1}^T c_r^t(x)  ~=~ \cO(\log n)}\).
\item If \(G\) is a graph of maximum degree \(\Delta\), then
\(\displaystyle{\sum_{t=1}^T c_r^t(x)  ~=~ \cO(n^{1-\varepsilon})}\), where \[\varepsilon = \left(\dfrac{\log(\Delta+2)}{\log(\Delta)}-1\right) > \dfrac{1}{\Delta \log (\Delta)}.\]
\end{enumerate}

\end{theorem}

\begin{proof}
We fist show (3), and then adapt la proof to show (1) and (2).  We define an energy operator relative to \(r\), giving weights to the edges of \(G\) that decrease exponentially with the distance from \(r\). Formally, we denote by \(E_{r}^t\) the \emph{energy operator centered in \(r\)}, defined as follows:

\[E_r^t(x)=\sum_{u,v\in V}\tilde{a}_{u,v}|x_u^{t+1}-x_v^{t}|\]
where \(\tilde{a}_{u,v} = a_{u,v}\cdot\alpha^{\delta(u,v)}\), with \(\alpha \in (0,1)\) a constant that we will fix later, and \(\delta(u,v)=\min\{d(r,u),d(r,v)\}\). Then,

\begin{align*}
    E_{r}^{t+1}(x) - E_r^{t}(x)
    =&\sum_{\{u,v\}\in E}\alpha^{\delta(u,v)}|x_u^{t+1}-x_v^{t}|-\sum_{\{u,v\}\in E}\alpha^{\delta(u,v)}|x_u^{t}-x_v^{t-1}|\\
    =&\sum_{\{u,v\}\in E}\alpha^{\delta(u,v)}|x_u^{t+1}-x_v^{t}|-\sum_{\{u,v\}\in E}\alpha^{\delta(u,v)}|x_u^{t-1}-x_v^{t}|\\
    =&\sum_{\{u,v\}\in E}\alpha^{\delta(u,v)}\left(|x_u^{t+1}-x_v^{t}|-|x_u^{t-1}-x_v^{t}|\right)
\end{align*}

We aim to upper bound \(E_{r}^{t+1}(x) - E_r^{t}(x)\). Observe that for \(u \in \partial B(r,i)\) and \(v\in N(u)\), the value of \(\delta(u,v)\) is either \(i-1\) or \({i}\). Suppose that \(c_u^t \neq 0\). We have that \(\left(|x_u^{t+1}-x_v^{t}|-|x_u^{t-1}-x_v^{t}|\right)\) is maximized when almost half of the neighbors of \(u\) are in a different state than \(u\) in \(t-1\), and exactly those neighbors are connected with edges of weight \(\alpha^{i-1}\). Therefore,

\begin{align*}
    E_{r}^{t+1}(x) - E_r^{t}(x) =&\sum_{i=0}^{\infty}\sum_{u\in \partial B(r,i)}\sum_{v\in N(u)}\alpha^{\delta(u,v)}\left(|x_u^{t+1}-x_v^{t}|-|x_u^{t-1}-x_v^{t}|\right)\\
    \leq&-c_r^t(x) \\&+ \sum_{i=1}^{\infty}\sum_{\substack{ u \in \partial B(r,i)\\ d(u) \textrm{ is even}}} c_u^t(x) \left(\left(\dfrac{d(u)}2 \right) \alpha^{i-1} - \left(\dfrac{d(u)}2+1\right)\alpha^i \right)\\
    &+ \sum_{i=1}^{\infty}\sum_{\substack{ u \in \partial B(r,i)\\ d(u) \textrm{ is odd}}} c_u^t(x) \left(\left(\dfrac{d(u)-1}2 \right) \alpha^{i-1} - \left(\dfrac{d(u)+1}2\right)\alpha^i \right)
\end{align*}

We now choose \(\alpha\). We  impose that for each \(u\in V\setminus \{r\}\) such that \(d(u)\) is even,
\begin{equation}\label{eq:even}\left(\dfrac{d(u)}2 \right) \alpha^{i-1} - \left(\dfrac{d(u)}2+1\right)\alpha^i \leq 0 \Rightarrow \alpha \geq \dfrac{d(u)}{d(u) +2};\end{equation}
and for each \(u\in V\setminus \{r\}\) such that \(d(u)\) is odd,

\begin{equation}\label{eq:odd}\left(\dfrac{d(u)-1}2 \right) \alpha^{i-1} - \left(\dfrac{d(u)+1}2\right)\alpha^i \leq 0 \Rightarrow \alpha \geq \dfrac{d(u)-1}{d(u)+1}.\end{equation}
Picking \(\alpha = \dfrac{\Delta}{\Delta+2}\) we obtain that conditions \ref{eq:even} and \ref{eq:odd} are satisfied for every \(u\in V \setminus \{r\}\). Then, \[E_{r}^{t+1}(x) - E_r^{t}(x) \leq -c_r^t.\] Using that \(E_r^t(x)\geq 0\) for every \(t>0\), we obtain

\[
\sum_{t=1}^T c_r^t(x) \leq \sum_{t=1}^T \left(E_r^{t}(x) - E_r^{t+1}(x)\right) = E_r^1(x) -  E_r^T(x) \leq E_r^1(x)
\]

To obtain our bound for \(\sum_{t=1}^T c_r^t(x) \), we upper bound  \(E_r^1(x)\). Observe that

\[ E_r^1(x) =  \sum_{\{u,v\}\in E}\alpha^{\delta(u,v)}|x_u^{t+1}-x_v^{t}| 
 =  \sum_{i=0}^\infty \sum_{u\in \partial B(r,i)} \sum_{v\in N(u)} \alpha^{\delta(u,v)}|x_u^{1}-x_v^{0}|\]

We have that the previous expression is maximized when, for each \(u\in V\), almost half of the neighbors \(v\) of \(u\) satisfy \(x_v^0 \neq x_u^1\), and the edge connecting such neighbors has the maximum possible weight. In that case, we obtain:

\begin{equation}\label{eq:bound}
 E_r^1(x)  \leq \dfrac{\Delta}2 \left(1+ \sum_{i=1}^\infty |\partial B(r,i)| \alpha^{i-1}\right)
\end{equation} 

 Now let us fix \(q = \dfrac{\log(n)}{\log \Delta}\). We have that:
\begin{align*}
E_r^1(x) & \leq      \dfrac{\Delta}2\left( 1 + \sum_{i=1}^q |\partial B(r,i)|   \alpha^{i-1} + \sum_{i=q+1}^\infty |\partial B(r,i)|  \alpha^{i-1}  \right) \\
& \leq \dfrac{\Delta+2}2\left(  \sum_{i=0}^q (\Delta  \alpha)^{i} + \sum_{i=q+1}^\infty n \alpha^{i}  \right)\\
& = \dfrac{\Delta+2}2\left(  \dfrac{(\Delta \alpha)^q -1}{\Delta\alpha -1} + (\alpha)^{q+1}n \dfrac{1}{1-\alpha}\right)\\
& = \left(\dfrac{(\Delta+2)^2}{2(\Delta+1)(\Delta - 2)}\right) \left((\Delta \alpha)^q -1 \right)+ \dfrac{\Delta(\Delta +2)}{4} n(\alpha)^q  
\end{align*}

Now observe that 

\[(\Delta \alpha)^q = 2^{q (2 \log \Delta - \log (\Delta + 2))} = n^{1 - \left(\frac{\log(\Delta+2)}{\log(\Delta)} - 1\right) } =  n^{1 - \varepsilon}\]
and 
\[n(\alpha)^q = n\cdot 2^{q (\log \Delta - \log (\Delta + 2))} = n^{1 - \left(\frac{\log(\Delta+2)}{\log(\Delta)} - 1\right) }= n^{1-\varepsilon}\]

We conclude that \[\sum_{t=1}^T c_r^t(x) \leq  E_r^1(x) \leq \left(\dfrac{(\Delta+2)^2}{2(\Delta+1)(\Delta - 2)} + \dfrac{\Delta(\Delta + 2)}{4}\right) \cdot n^{1-\varepsilon}= \cO(n^{1-\varepsilon})\]

This finishes the proof of (3).

To prove (1), consider Equation~\ref{eq:bound} and observe that in this case:

\[E_r^1(x) \leq  \dfrac{\Delta}2 \left( 1+ \sum_{i=1}^\infty f(i) \alpha^{i-1}\right) \]

Since \(f(i)\) is sub-exponential, we have that there exists a large enough \(\ell>0\) such that \[f(\ell) \leq  \left(\dfrac{\Delta+1}{\Delta}\right)^{\ell} .\] 

Then, 
\begin{align*}
E_r^1(x) &\leq  \dfrac{\Delta}2 \left( 1+ \sum_{i=1}^\infty f(i) \alpha^{i-1}\right) \\
&=\dfrac{\Delta}2 \left(1+\sum_{i=1}^{\ell} f(i) \alpha^{i-1} + \sum_{i=i^*+1}^{\infty} f(i) \alpha^{i-1} \right)\\
& \leq \dfrac{\Delta}2+ \dfrac{\Delta+2}2 \left(\left(\dfrac{\Delta+1}{\Delta}\right)^{\ell} \cdot \sum_{i=1}^{\ell} \left(\dfrac{\Delta}{\Delta + 2} \right)^{i} + \sum_{i=\ell+1}^{\infty} \left(\dfrac{\Delta+1}{\Delta+2}\right)^{i} \right)\\
&=\dfrac{\Delta}2+ \dfrac{\Delta(\Delta+2)}4\cdot\left(\dfrac{\Delta+1}{\Delta}\right)^{\ell} \cdot \left(1- \left(\dfrac{\Delta}{\Delta + 2} \right)^{\ell} \right) + \dfrac{(\Delta+2)^2}2 \left(\dfrac{\Delta+1}{\Delta+2}\right)^{\ell+1}
\end{align*}
Since \(\Delta \leq f(1)\), we deduce that 

\[\sum_{t=1}^T c_r^t(x) \leq  E_r^1(x) = \cO(1)\]

Finally,  let us prove (2). In a graph of maximum degree 3 we can pick \(\alpha = 1/2\) and satisfy, for every \(u\in V\),  Equations~\ref{eq:even} and \ref{eq:odd}. Moreover, for every \(v\in V\) and every \(i\geq 0\), \(|\partial B(v,i)| \leq 3 \cdot 2^{i-1}\). Then, starting from Equation~\ref{eq:bound} and picking \(q = 2\log(n)\) we have that:

\begin{align*}
E_r^1(x) & \leq 1 + \sum_{i=1}^q |\partial B(r,i)|   \alpha^{i-1} + \sum_{i=q+1}^\infty |\partial B(r,i)|  \alpha^{i-1}   \\
& \leq  1+ 3\sum_{i=1}^q 2^{i-1}  \alpha^{i-1} + \sum_{i=q+1}^\infty n \alpha^{i} \\
& = 1 + 3q  + n(\alpha)^{q} = 6 \log n + 1 +\dfrac1n
\end{align*}
We deduce that
\[\sum_{t=1}^T c_r^t(x) \leq  E_r^1(x) = \cO(\log n)\]

\end{proof}

\section{Certification Upper-bounds}

In this section we give protocols for the certification of \electionpred\ and \electionpredi. 

\subsection{Upper-bound for~\electionpred}

For a graph \(G\) let us define define \(\changes(G) = \max_{x}\left(\max_{v} \sum_{t>0}c^t_v(x)\right)\). For a set of graphs \(\mathcal{G}\) we define \(\changes(\mathcal{G}) = \max_{G \in \mathcal{G}} \changes(G)\). Given an infinite set of graphs \(\mathcal{G}\) and \(n>0\), we denote by \(\mathcal{G}_n\) the subset of \(\mathcal{G}\) of size \(n\).

\begin{lemma}\label{theo:pls}
 For each \(n\)-node graph \(G\) there is a proof labeling scheme for problem \(\electionpred\) with certificates of size \(\cO(\changes(G) \cdot \log n)\).
\end{lemma}

\begin{proof}
Let \((G,x,T)\) be an instance of the problem \(\electionpred\). Let \(N\) be an upper-bound on the size of \(G\) that is initially known by all vertices. The certificate of  node \(v\in V(G)\) consists in a pair 
 \[\textsf{Certificate}(v) = (\textsf{evenChanges}(v), \textsf{oddChanges}(v))\]
 where:

 \begin{itemize}
\item \(\textsf{evenChanges}(v)\) is a set of pairs \((q,t)\) where \(t \in [N^2]\) is an even time-step such that \(x_v(t+2) = q \neq x_v(t)\). The set \(\textsf{evenChanges}(v)\) also includes the pair \((x_v, 0)\).
\item \(\textsf{oddChanges}(v)\)  is a set of pairs \((q,t)\) where \(t \in [N^2]\) is an odd time-step such that \(x_v(t+2) = q \neq x_v(t)\). The set \(\textsf{oddChanges}(v)\) also includes the pair \((q, 1)\), where \(q\) represents the sate of \(v\) in time-step \(1\). 
 \end{itemize}

\medskip

\noindent{\bf Verification Algorithm.} In the verification round, node \(v\) receives \(\textsf{Certificate}(u)\) for every~\(u\in N(v)\). For each \(u\in N[v]\), vertex \(v\) computes the  vector \(\textsf{orbit}(u)\)  representing the orbit \(u\). Formally, \(v\) computes the vector \(\textsf{orbit}(u) \in \{0,1\}^{N^2+1}\) defined as follows. For each \(i \in \{0, \dots, N^2\}\), we define:
\[t_{\textrm{inf}}(i) = \left\{
\begin{array}{ll} \max \{ t \in \{0, \dots, i\}: \exists q \in \{0,1\} \textrm{ s.t } (t,q) \in \textsf{evenChanges}(u)  & \textrm{if } i \textrm{ is even,}\\
\max \{ t \in \{0, \dots, i\}: \exists q \in \{0,1\} \textrm{ s.t } (t,q) \in \textsf{oddChanges}(u)   & \textrm{if } i \textrm{ is odd.}
\end{array} \right.\]
Then, \(\textsf{orbit}(u)_i = q\) where \(q\) is the state such that \((q,t_{\textrm{inf}}(i)) \in \textsf{evenChanges}(v)\) when \(i\) is even, or where \((q,t_{\textrm{inf}}(i)) \in \textsf{oddChanges}(v)\) when \(i\) is odd. Vertex \(v\) rejects \(\textsf{Certificate}(v)\) if the orbit of \(v\) does not coincide with the majority function on its closed neighborhood. Formally, vertex \(v\) checks that for every \(t \in [N^2]\) the value of \(\textsf{orbit}_t(v) = 1\) if and only if \(\sum_{u \in S(v)} \textsf{orbit}_{t-1}(u) > d(v)/2\), where \(S(v) = N(v)\) when the degree of \(v\) is odd, and \(S(v) = N[v]\) when the degree of \(v\) is even. Finally \(v\) accepts when \(\textsf{orbit}_T(v) = y_v\).\\

\noindent{\bf Completeness and Soundness.} Let us analyze now the completeness and soundness of our proof labeling scheme. 

\medskip

{\it Completeness.} Suppose first that \((G,x,T)\) belongs to \(\electionpred\). Then, we can choose \(\textsf{Certificate}(v)\) as described above, making  every node in \(G\) accept. 

\medskip

{\it Soundness.} Now let us suppose that every node in \(G\) has accepted a given certificate. Following the verification algorithm, \(v\) is capable of reconstruct the orbit of every node \(u\in N[v]\). Then, all the neighbors of vertex \(v\) agree in the same orbit of \(v\). Therefore, if every vertex did not reject in the verification of the orbits, we deduce that \(\textsf{orbit}_t(v) = x_v^t\) for every \(t\in [N^2]\). In particular, we have that \(x^T = y\). \\

\noindent{\bf Size of the Certificates.} For each vertex \(v\), \(\textsf{Certificate}(v)\) can be encoded in \(\cO(\textsf{Changes}(G) \cdot \log n)\) bits. Indeed, each pair \[(q,t) \in \textsf{evenChanges}(v) \cup \textsf{oddChanges}(v)\] can be encoded in \(1 + \log (N^2) = \cO(\log n)\) bits. Moreover, the cardinality of \( \textsf{evenChanges}(v) \cup \textsf{oddChanges}(v)\) is, by definition, at most \(\textsf{Changes}(G)\).
\end{proof}

Theorem~\ref{theo:boundedchange} pipelined with Lemma~\ref{theo:pls}  gives the main result of this section. 

\begin{theorem}\label{coro:main}~

\begin{enumerate}
 \item There is a 1-round proof labeling scheme for \(\electionpred\) with certificates on \(\cO(\log n)\) bits in \(n\)-node networks of sub-exponential growth.
\item There is a 1-round proof labeling scheme for \(\electionpred\) with certificates on \(\cO(\log^2 n)\) bits in \(n\)-node networks of maximum degree \(3\).

\item There is a 1-round proof labeling scheme for \(\electionpred\) with certificates on \(\cO(n^{1-\epsilon}\log n)\) bits in \(n\)-node networks of maximum degree \(\Delta>2\), where \(\varepsilon = {1}/{\Delta \log (\Delta)}\).

\end{enumerate}

\end{theorem}

\subsection{Upper-bound for~\electionpredi}

We now show the proof labeling schemes for \electionpredi. Our bounds of the size of the certificates is obtained from Theorem~\ref{coro:main} and the following result. Let us define \(\textsc{Count-Ones}\) as the problem of deciding, given a configuration \(x\) and a constant \(k\), if in the graph there are exactly \(k\) nodes in state \(1\). Formally, 

\[\textsc{Count-Ones} = \left\{(G,(x,k)): x: V \rightarrow \{0,1\} ,~ k\geq 0 ~\textrm{and}~ \sum_{v\in V(G)} x_v = k\right\} \]

In \cite{KormanKP10} it is shown that there is a PLS for \textsc{Count-Ones} with certificates of size \(\cO(\log n)\). 

\begin{lemma}\label{lem:count-ones}(see \cite{KormanKP10})
There is a proof labeling scheme for \(\textsc{Count-Ones}\) with certificates of size~\(\cO(\log n)\). 
\end{lemma}

Roughly, the idea of the PLS of Lemma \ref{lem:count-ones} the following: the certificate of a node \(v\) is a tuple \((\textsf{root}, \textsf{parent}, \textsf{distance}, \textsf{count})\) where \(\textsf{root}\) is the identifier of the root of a rooted spanning tree \(\tau\) of \(G\), \(\textsf{parent}\) is the identifier of the parent of \(v\) in \(\tau\), \(\textsf{distance}\) is the distance of \(v\) to the root and \( \textsf{count}\) is the number of nodes in state \(1\) in the subgraph of \(G\) induced by the descendants of \(v\) in \(\tau\). Then, every vertex checks the local coherence of the certificates, and the root also checks whether \(\textsf{count}\) equals \(k\). The upper bounds for \electionpredi\ follow directly from Theorem \ref{coro:main} and Lemma \ref{lem:count-ones}.

\clearpage
\begin{theorem}\label{coro:electionpredi}~
\begin{enumerate}
 \item There is a 1-round proof labeling scheme for \(\electionpredi\) with certificates on \(\cO(\log n)\) bits in \(n\)-node networks of sub-exponential growth.
\item There is a 1-round proof labeling scheme for \(\electionpredi\) with certificates on \(\cO(\log^2 n)\) bits in \(n\)-node networks of maximum degree \(3\).

\item There is a 1-round proof labeling scheme for \(\electionpredi\) with certificates on \(\cO(n^{1-\epsilon}\log n)\) bits in \(n\)-node networks of maximum degree \(\Delta>2\), where \(\varepsilon = {1}/{\Delta \log (\Delta)}\).

\end{enumerate}

\end{theorem}

\begin{proof}
Given an instance \(G,x, T\) of \electionpredi, the protocol consists in giving each node \(v\)
\begin{enumerate}
\item The state \(y_v\) of \(v\) on time-step \(T\). 
\item An integer \(n_v\) representing the number of nodes in \(G\)
\item An integer \(p_v\) representing \(\sum_{v\in V} y_v\).
\item The certification of \electionpred\ of instance \((G,x,(y_v)_{v\in V},T)\).
\item The certification of \(\textsc{Count-Ones}\) on instance \((G,(1)_{v \in V},n_v)\).
\item The certification of \(\textsc{Count-Ones}\) on instance \((G,(y)_{v \in V},p_v)\)
\end{enumerate}
In the verification round \(v\) simulates the verification round of the PLS for \electionpred\ and \(\textsc{Count-Ones}\) using the corresponding certificates, and rejects if any of the simulations rejects. Finally, \(v\) accepts if and only if \(p_v/n_v > 1/2\). The completeness and soundness of the protocol follow from Theorem~\ref{theo:pls} and Lemma~\ref{lem:count-ones}. The bound on the size of the certificates is obtained by Theorem~\ref{theo:boundedchange}.
\end{proof}

\section{Lower-bounds}

We first prove that every locally checkable proof for problems \(\electionpred\) or \(\electionpredi\) on arbitrary \(n\)-node graphs require certificates of size \(\Omega(n)\). The proof is a reduction from the disjointedness problem in non-deterministic communication complexity. In this problem, Alice receives a vector \(a\in \{0,1\}^n\) and Bob a vector \(b \in \{0,1\}^n\). The players can perform a series of communication rounds and have the task of deciding whether there exists a coordinate \(i\in \{1,\dots,n\}\) such that \(a_i = b_i\). In \cite{kushilevitz1997communication} it is shown that the non-deterministic communication complexity of this problem is \(\Omega(n)\).

\begin{theorem}\label{theo:lowerbound}
Every locally checkable proof certifying  \(\electionpredi\) or \(\electionpred\) on arbitrary \(n\)-node graphs has certificates on \(\Omega(n)\) bits. 
\end{theorem}

The construction has three main parts: one produces the sequence of bits of Alice at some node, one produces the sequence of bits of Bob at some other node, and the last part is used to change the majority of the entire graph forever as soon as Alice and Bob have a 1 in their sequence at a common position.

The parts producing the sequence are called \emph{sequencer gadgets} and they are using high degree nodes. Such sequencer gadget are based on building blocks called \emph{timer gadgets}.
A timer gadget is a subgraph that enforces a state change $0\to 1$  or \(1\to 0\) at some nodes,  at specific times, independently of the context to which these nodes are connected (provided they have a single neighbor outside the gadget). 


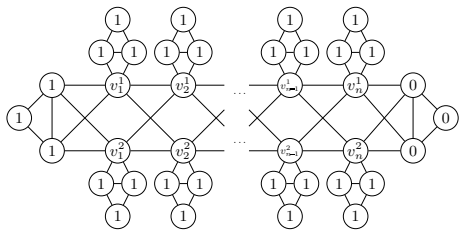
\begin{figure}
\centering
\resizebox*{0.5\textwidth}{!}{\tikzstyle{vertice}=[scale=1.5,draw,circle,minimum size=0.5cm,inner sep=0pt]
\begin{tikzpicture}

\node[vertice] (v29) at (-3.5,2.5) {$v^1_1$};
\node [vertice] (v27) at (-1.5,2.5) {$v^1_2$};

\node [] at (0.25,2.25) {$\cdots$};
\node [scale=1.0,draw,circle,minimum size=0.75cm,inner sep=0pt] (v33) at (1.75,2.5) {$v^1_{n\shortminus1}$};
\node [vertice] (v32) at (3.75,2.5) {$v^1_n$};

\node [vertice] (v1) at (-4,3.5) {1};
\node [vertice] (v3) at (-3,3.5) {1};
\node [vertice] (v2) at (-3.5,4.5) {1};
\node [vertice] (v4) at (-2,3.5) {1};
\node [vertice] (v5) at (-1,3.5) {1};
\node [vertice] (v6) at (-1.5,4.5) {1};

\node [vertice] (v7) at (1.25,3.5) {1};
\node [vertice] (v9) at (2.25,3.5) {1};
\node [vertice] (v8) at (1.75,4.5) {1};
\node [vertice] (v10) at (3.25,3.5) {1};
\node [vertice] (v12) at (4.25,3.5) {1};
\node [vertice] (v11) at (3.75,4.5) {1};

\node [vertice] (v25) at (-5.5,2.5) {1};
\node [vertice] (v0) at (-6.5,1.5) {1};

\node[vertice] (v26) at (-3.5,0.5) {$v^2_1$};
\node [vertice] (v30) at (-1.5,0.5) {$v^2_2$};

\node [] at (0.25,0.75) {$\cdots$};

\node [scale=1.0,draw,circle,minimum size=0.75cm,inner sep=0pt] (v31) at (1.75,0.5) {$v^2_{n\shortminus1}$};
\node [vertice] (v34) at (3.75,0.5) {$v^2_n$};
\node [vertice] (v35) at (5.5,2.5) {0};
\node [vertice] (v36) at (5.5,0.5) {0};
\node [vertice] (v37) at (6.5,1.5) {0};

\node [vertice] (v13) at (-4,-0.5) {1};
\node [vertice] (v14) at (-3,-0.5) {1};
\node [vertice] (v15) at (-3.5,-1.5) {1};
\node [vertice] (v16) at (-2,-0.5) {1};
\node [vertice] (v17) at (-1,-0.5) {1};
\node [vertice] (v18) at (-1.5,-1.5) {1};

\node [vertice] (v19) at (1.25,-0.5) {1};
\node [vertice] (v20) at (2.25,-0.5) {1};
\node [vertice] (v21) at (1.75,-1.5) {1};
\node [vertice] (v22) at (3.25,-0.5) {1};
\node [vertice] (v23) at (4.25,-0.5) {1};
\node [vertice] (v24) at (3.75,-1.5) {1};

\node [vertice] (v28) at (-5.5,0.5) {1};
\draw (v0) edge (v28);
\draw (v0) edge (v25);
\draw (v25) edge (v28);
\draw  (v1) edge (v2);
\draw  (v2) edge (v3);
\draw  (v1) edge (v3);
\draw  (v4) edge (v5);
\draw  (v6) edge (v5);
\draw  (v4) edge (v6);
\draw (v7) -- (v8) -- (v9) -- (v7);
\draw (v10) -- (v11);
\draw (v11) -- (v12) -- (v10);
\draw (v35) edge (v36);
\draw (v37) edge (v36);
\draw (v35) edge (v37);
\draw (v35) edge (v32);
\draw (v36) edge (v34);
\draw (v35) edge (v34);
\draw (v36) edge (v32);

\draw (v13) -- (v14) -- (v15) -- (v13);
\draw (v16) -- (v17) -- (v18) -- (v16);
\draw (v19) -- (v20) -- (v21) -- (v19);
\draw (v22) -- (v23) -- (v24) -- (v22);
\draw (v25) -- (v26) -- (v27) -- (-0.25,1.25) node [] {};
\draw (v28) -- (v29) -- (v30) -- (-0.25,1.75) node [] {};
\draw (v25) -- (v29) -- (v27) -- (-0.25,2.5) node [] {};
\draw (v28) -- (v26) -- (v30) -- (-0.25,0.5) node [] {};
\draw (v29) -- (v1);
\draw (v3) -- (v29);
\draw (v4) -- (v27) -- (v5);
\draw (v13) -- (v26) -- (v14);
\draw (v16) -- (v30) -- (v17);
\draw (0.5,1.75) node [] {} -- (v31) -- (v32);
\draw (0.5,1.25) node [] {} -- (v33) -- (v34);
\draw (v31) -- (v34);
\draw (v33);
\draw (v32);
\draw (v33) -- (v32);
\draw (0.5,2.5) node [] {} -- (v33);
\draw (0.5,0.5) node [] {} -- (v31);

\draw (v7);
\draw (v7) -- (v33) -- (v9);
\draw (v10) -- (v32) -- (v12);
\draw (v19) -- (v31) -- (v20);
\draw (v22) -- (v34) -- (v23);
\end{tikzpicture}}
\caption{\label{fig:timergadget}Timer gadget from Lemma~\ref{lem:timergadget}.} 
\end{figure}

\begin{lemma}[Timer gadget]\label{lem:timergadget}
  For any $n$, there exists a graph $G=(V,E)$ of size ${O(n)}$ with $2n$ distinguished nodes ${\{v^1_t,v^2_t: 1\leq t\leq n\}\subseteq V}$ and an initial configuration ${x\in \{0,1\}^V}$ for $G$ such that, for any graph $H=(V',E')$ containing $G$ as induced subgraph and ${y\in \{0,1\}^{V'}}$ an initial configuration for $H$, if
  \begin{itemize}
  \item ${y_{|V} = x}$ and
  \item nodes ${v_t^j}$ have at most one neighbor outside $V$ in $H$ and no other node of $V$ has a neighbor outside $V$,
  \end{itemize}
  then the orbit ${(y^t)_{t\geq 0}}$ of configurations of $H$ starting from $y=y^0$ verifies
  \[y^t_{v^j_k} =
    \begin{cases}
      0 &\text{ if ${t< k}$,}\\
      1 &\text{ else,}
    \end{cases}
  \]
  for ${1\leq k\leq n}$ and ${j\in\{1,2\}}$.
\end{lemma}
\begin{proof}
  Consider for $G$ the graph of Figure~\ref{fig:timergadget} with distinguished nodes $v_i^j$ and initial configuration $x$ where all distinguished nodes are in state $0$ and all other nodes are in the state specified by the figure.
  Note first that in $x$, all nodes which are not among the $v_i^j$ belongs to a triangle of nodes in the same state, and have at most two neighbors outside the triangle.
  Therefore they maintain a local majority corresponding to their own state and will never change of state (recall that by hypothesis they don't have neighbors outside $G$).
  In $x$, all distinguished nodes have $4$ neighbors in $G$ in state $0$, and $2$ neighbors in $G$ in state $1$, except $v_1^1$ and $v_1^2$ with respectively $2$ for state $0$ and $4$ for state $1$.
  Therefore, at step $1$ and independently of the state of possible neighbors outside $G$ (at most one for each $v_i^j$), nodes $v_1^1$ and $v_1^2$ turn into state $1$ and all other distinguished nodes stay in state $0$.
  With the same reasoning, it is straightforward to show by induction that, at step $t$, node $v_i^j$ is $0$ if and only if $t<i$.
\end{proof}

Of course a symmetric timer gadget that triggers $1\to 0$ state changes can be obtained the same way.
A sequencer gadget is a subgraph that enforces an arbitrary sequence of states at some node, independently of the context to which this node is connected (provided it has a single neighbor outside the gadget). 

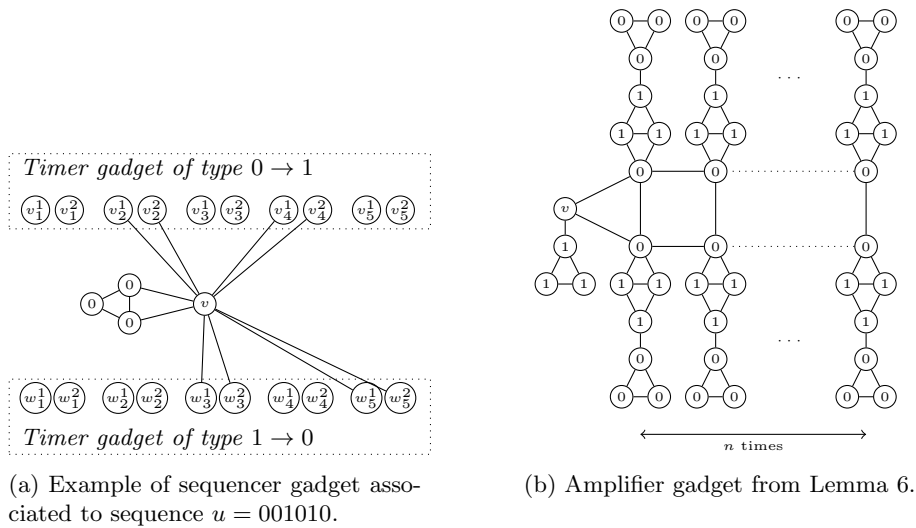
\begin{figure}
    \centering

\begin{subfigure}[t]{0.45\textwidth}
  \centering
      \tikzstyle{vertice}=[scale=1,draw,circle,minimum size=0.3cm,inner sep=0pt]
  \begin{tikzpicture}\tiny
    \begin{scope}[shift={(.5,0)}]
      \node[vertice] (v1) at (0,0) {$v$};
      \node[vertice] (v2) at (-1,.25) {$0$};
      \node[vertice] (v3) at (-1,-.25) {$0$};
      \node[vertice] (v4) at (-1.5,0) {$0$};
      \draw (v1) edge (v2);
      \draw (v1) edge (v3);
      \draw (v3) edge (v2);
      \draw (v4) edge (v2);
      \draw (v3) edge (v4);
    \end{scope}
    \draw[dotted] (-2.1,1) rectangle (3.5,2);
    \draw (0,1.8) node {\small\it Timer gadget of type $0\to 1$};
    \node[vertice] (v2) at (-1.75,1.25) {$v_1^1$};
    \node[vertice] (v3) at (-1.3,1.25) {$v_1^2$};
    \begin{scope}[shift={(2.2,0)}]
      \node[vertice] (v2) at (-1.75,1.25) {$v_3^1$};
      \node[vertice] (v3) at (-1.3,1.25) {$v_3^2$};
    \end{scope}
    \begin{scope}[shift={(4.4,0)}]
      \node[vertice] (v2) at (-1.75,1.25) {$v_5^1$};
      \node[vertice] (v3) at (-1.3,1.25) {$v_5^2$};
    \end{scope}
    \begin{scope}[shift={(1.1,0)}]
      \node[vertice] (v2) at (-1.75,1.25) {$v_2^1$};
      \node[vertice] (v3) at (-1.3,1.25) {$v_2^2$};
      \draw (v1) edge (v2);
      \draw (v1) edge (v3);
    \end{scope}
    \begin{scope}[shift={(3.3,0)}]
      \node[vertice] (v2) at (-1.75,1.25) {$v_4^1$};
      \node[vertice] (v3) at (-1.3,1.25) {$v_4^2$};
      \draw (v1) edge (v2);
      \draw (v1) edge (v3);
    \end{scope}
    \draw[dotted] (-2.1,-1) rectangle (3.5,-2);
    \draw (0,-1.8) node {\small\it Timer gadget of type $1\to 0$};
    \node[vertice] (v2) at (-1.75,-1.25) {$w_1^1$};
    \node[vertice] (v3) at (-1.3,-1.25) {$w_1^2$};
    \begin{scope}[shift={(1.1,-0)}]
      \node[vertice] (v2) at (-1.75,-1.25) {$w_2^1$};
      \node[vertice] (v3) at (-1.3,-1.25) {$w_2^2$};
    \end{scope}
    \begin{scope}[shift={(3.3,-0)}]
      \node[vertice] (v2) at (-1.75,-1.25) {$w_4^1$};
      \node[vertice] (v3) at (-1.3,-1.25) {$w_4^2$};
    \end{scope}
    \begin{scope}[shift={(2.2,-0)}]
      \node[vertice] (v2) at (-1.75,-1.25) {$w_3^1$};
      \node[vertice] (v3) at (-1.3,-1.25) {$w_3^2$};
      \draw (v1) edge (v2);
      \draw (v1) edge (v3);
    \end{scope}
    \begin{scope}[shift={(4.4,-0)}]
      \node[vertice] (v2) at (-1.75,-1.25) {$w_5^1$};
      \node[vertice] (v3) at (-1.3,-1.25) {$w_5^2$};
      \draw (v1) edge (v2);
      \draw (v1) edge (v3);
    \end{scope}
  \end{tikzpicture}
  
  
      \caption{Example of sequencer gadget associated to sequence ${u=001010}$.}
  \label{fig:sequencergadget}
\end{subfigure}\hfill
\begin{subfigure}[t]{0.45\textwidth}
  \centering
  \tikzstyle{vertice}=[scale=1,draw,circle,minimum size=0.3cm,inner sep=0pt]
  \begin{tikzpicture}\tiny
    \node[vertice] (v1) at (0,0) {$v$};
    \node[vertice] (v2) at (0,-.5) {$1$};
    \node[vertice] (v3) at (.25,-1) {$1$};
    \node[vertice] (v4) at (-.25,-1) {$1$};
    \draw (v1) edge (v2);
    \draw (v3) edge (v2);
    \draw (v4) edge (v2);
    \draw (v3) edge (v4);
    \node[vertice] (v2) at (1,.5) {$0$};
    \node[vertice] (v3) at (1,-.5) {$0$};
    \draw (v1) edge (v2);
    \draw (v1) edge (v3);
    \draw (v2) edge (v3);
    \node[vertice] (v4) at (1.25,-1) {$1$};
    \node[vertice] (v5) at (.75,-1) {$1$};
    \node[vertice] (v6) at (1,-1.5) {$1$};
    \draw (v4) edge (v3);
    \draw (v5) edge (v3);
    \draw (v4) edge (v6);
    \draw (v5) edge (v6);
    \draw (v4) edge (v5);
    \node[vertice] (v4) at (1.25,-2.5) {$0$};
    \node[vertice] (v5) at (.75,-2.5) {$0$};
    \node[vertice] (v7) at (1,-2) {$0$};
    \draw (v4) edge (v7);
    \draw (v4) edge (v5);
    \draw (v5) edge (v7);
    \draw (v6) edge (v7);
    \node[vertice] (v4) at (1.25,1) {$1$};
    \node[vertice] (v5) at (.75,1) {$1$};
    \node[vertice] (v6) at (1,1.5) {$1$};
    \draw (v4) edge (v2);
    \draw (v5) edge (v2);
    \draw (v4) edge (v6);
    \draw (v5) edge (v6);
    \draw (v4) edge (v5);
    \node[vertice] (v4) at (1.25,2.5) {$0$};
    \node[vertice] (v5) at (.75,2.5) {$0$};
    \node[vertice] (v7) at (1,2) {$0$};
    \draw (v4) edge (v7);
    \draw (v4) edge (v5);
    \draw (v5) edge (v7);
    \draw (v6) edge (v7);
    \begin{scope}[shift={(1,0)}]
          \node[vertice] (vv2) at (1,.5) {$0$};
          \node[vertice] (vv3) at (1,-.5) {$0$};
          \draw (v2) edge (vv2);
          \draw (v3) edge (vv3);
          \draw (vv2) edge (vv3);
          \node[vertice] (v4) at (1.25,-1) {$1$};
          \node[vertice] (v5) at (.75,-1) {$1$};
          \node[vertice] (v6) at (1,-1.5) {$1$};
          \draw (v4) edge (vv3);
          \draw (v5) edge (vv3);
          \draw (v4) edge (v6);
          \draw (v5) edge (v6);
          \draw (v4) edge (v5);
          \node[vertice] (v4) at (1.25,-2.5) {$0$};
          \node[vertice] (v5) at (.75,-2.5) {$0$};
          \node[vertice] (v7) at (1,-2) {$0$};
          \draw (v4) edge (v7);
          \draw (v4) edge (v5);
          \draw (v5) edge (v7);
          \draw (v6) edge (v7);
          \node[vertice] (v4) at (1.25,1) {$1$};
          \node[vertice] (v5) at (.75,1) {$1$};
          \node[vertice] (v6) at (1,1.5) {$1$};
          \draw (v4) edge (vv2);
          \draw (v5) edge (vv2);
          \draw (v4) edge (v6);
          \draw (v5) edge (v6);
          \draw (v4) edge (v5);
          \node[vertice] (v4) at (1.25,2.5) {$0$};
          \node[vertice] (v5) at (.75,2.5) {$0$};
          \node[vertice] (v7) at (1,2) {$0$};
          \draw (v4) edge (v7);
          \draw (v4) edge (v5);
          \draw (v5) edge (v7);
          \draw (v6) edge (v7);
        \end{scope}
        \begin{scope}[shift={(3,0)}]
          \node[vertice] (v2) at (1,.5) {$0$};
          \node[vertice] (v3) at (1,-.5) {$0$};
          \draw (v2) edge[dotted] (vv2);
          \draw (v3) edge[dotted] (vv3);
          \draw (v2) edge (v3);
          \node[vertice] (v4) at (1.25,-1) {$1$};
          \node[vertice] (v5) at (.75,-1) {$1$};
          \node[vertice] (v6) at (1,-1.5) {$1$};
          \draw (v4) edge (v3);
          \draw (v5) edge (v3);
          \draw (v4) edge (v6);
          \draw (v5) edge (v6);
          \draw (v4) edge (v5);
          \node[vertice] (v4) at (1.25,-2.5) {$0$};
          \node[vertice] (v5) at (.75,-2.5) {$0$};
          \node[vertice] (v7) at (1,-2) {$0$};
          \draw (v4) edge (v7);
          \draw (v4) edge (v5);
          \draw (v5) edge (v7);
          \draw (v6) edge (v7);
          \node[vertice] (v4) at (1.25,1) {$1$};
          \node[vertice] (v5) at (.75,1) {$1$};
          \node[vertice] (v6) at (1,1.5) {$1$};
          \draw (v4) edge (v2);
          \draw (v5) edge (v2);
          \draw (v4) edge (v6);
          \draw (v5) edge (v6);
          \draw (v4) edge (v5);
          \node[vertice] (v4) at (1.25,2.5) {$0$};
          \node[vertice] (v5) at (.75,2.5) {$0$};
          \node[vertice] (v7) at (1,2) {$0$};
          \draw (v4) edge (v7);
          \draw (v4) edge (v5);
          \draw (v5) edge (v7);
          \draw (v6) edge (v7);
        \end{scope}
        \draw (3,1.75) node {$\cdots$};
        \draw (3,-1.75) node {$\cdots$};
        \draw[<->] (1,-3) -- node[midway,below] {$n$ times} (4,-3);
      \end{tikzpicture}
  \caption{Amplifier gadget from Lemma~\ref{lem:ampligadget}.}
  \label{fig:amplifiergadget}
\end{subfigure}
\caption{Key gadgets for proof of Theorem~\ref{theo:lowerbound}}
\end{figure}

\begin{lemma}[Sequencer gadget]\label{lem:seqgadget}
  For any $n$ and any sequence ${u\in\{0,1\}^n}$, there exists a graph ${G=(V,E)}$ of size ${O(n)}$ with one distinguished node ${v\in V}$ and an inital configuration ${x\in\{0,1\}^V}$ such that, for any graph $H=(V',E')$ containing $G$ as induced subgraph and ${y\in \{0,1\}^{V'}}$ an initial configuration for $H$, if
  \begin{itemize}
  \item ${y_{|V} = x}$ and
  \item node ${v}$ has at most one neighbor outside $V$ in $H$ and no other node of $V$ has a neighbor outside $V$,
  \end{itemize}
  then the orbit ${(y^t)_{t\geq 0}}$ of configurations of $H$ starting from $y=y^0$ verifies ${y^0_v = u_1}$, and ${y^t_v = u_t}$ for $1\leq t\leq n$, and $y^t_v = u_n$ for $t>n$.
\end{lemma}
\begin{proof}
  Suppose that ${u_1=0}$ (the other case is symmetric).
  Let $$1\leq s_1<s_2<\cdots<s_k<n$$ be the positions $j$ in $u$ such that ${u_{j+1}\neq u_j}$ and $u_j=0$.
  Symmetrically, let ${1\leq t_1<t_2<\cdots<t_l<n}$ be the positions $j$ in $u$ such that ${u_{j+1}\neq u_j}$ and $u_j=1$.
  Note that since we suppose ${u_1=0}$, it holds that ${1\leq s_1<t_1<s_2<t_2<\cdots}$ and either $k=l$ (if $u_n=0$) or $k=l+1$ (if $u_n=1$).
  We now describe the sequencer gadget associated to $u$. It is made of a first timer gadget that triggers $0\to 1$ state changes with distinguished nodes $v_i^j$ for ${j\in\{1,2\}}$ and ${1\leq i\leq n}$, and a second timer gadget that triggers $1\to 0$ state changes with distinguished nodes $w_i^j$ for ${j\in\{1,2\}}$ and ${1\leq i\leq n}$.
  Then, an additional node $v$ (the distinguished node of the sequencer gadget) is connected to timer gadgets as follows:
  \begin{itemize}
  \item for all ${1\leq i\leq k}$, $v$ is connected to both $v_{s_i}^1$ and $v_{s_i}^2$ ;
  \item for all ${1\leq i\leq l}$, $v$ is connected to both $w_{t_i}^1$ and $w_{t_i}^2$.
  \end{itemize}
  Moreover, if $k=l$, we add a clique of $3$ nodes in state $0$, both of which are connected to $v$.
  The initial configuration of the gadget is made by the initial configurations of each timer gadget (Lemma~\ref{lem:timergadget}) and all other nodes ($v$ and possibly the additional clique) in state $0$.
See Figure~\ref{fig:sequencergadget} for an example with $k=l$.
  
  Let us consider the behavior of this initialized gadget inside a larger graph $H$ as in the hypothesis of the lemma.
  Initially, $v$ is in state $0$, and it has $l+2$ neighbors in state $0$ and $l$ in state $1$ within the sequencer gadget, so even if $v$ has an additional neighbor outside the gadget, the majority at node $v$ is guaranteed to be $0$.
  By the behavior of timer gadgets, the neighboring configuration of node $v$ inside the gadget does not change until time $s_1$, when both $v_{s_1}^1$ and $v_{s_1}^2$ have turned into state $1$.
  So, at time $s_1$, node $v$ has $l$ neighbors in state $0$ and $l+2$ in state $1$ within the gadget.
  Therefore, independently of the potential additional neighbor outside the gadget and the parity of the degree of $v$, the local majority at $v$ is $1$ so $v$ turns into $1$ at step $s_1+1$.

  It is straightforward to prove by induction with the same analysis that the state of $v$ at step $j$ is $u_j$ for ${1\leq j\leq n}$ and $u_n$ for ${j>n}$ and $0=u_1$ at initial step.
\end{proof}

An amplifier gadget is subgraph with a special node, such that if at any time the special node has its two neighbors outside the subgraph in state 1, then the entire subgraph changes from an initial large majority of 0s to a steady large majority of 1s.

\begin{lemma}[Amplifier gadget]\label{lem:ampligadget}
  There is a constant ${1/2<\alpha\leq 1}$ and, for any large enough $n$, a graph ${G=(V,E)}$ of size ${\Omega(n)}$ with one distinguished node ${v\in V}$ and an inital configuration ${x\in\{0,1\}^V}$ such that for any graph $H=(V',E')$ containing $G$ as induced subgraph and ${y\in \{0,1\}^{V'}}$ an initial configuration for $H$, if
  \begin{itemize}
  \item ${y_{|V} = x}$ and
  \item $v$ has two neighbors outside $V$ in $H$ and no other node of $V$ has a neighbor outside $V$,
  \end{itemize}
  then the orbit ${(y^t)_{t\geq 0}}$ of configurations of $H$ starting from $y=y^0$ verifies:
  \begin{itemize}
  \item if node $v$ has its two neighbors outside $V$ in state $1$ at some time step, then the proportion of 0s among nodes of $V$ becomes at most ${1-\alpha}$ after time $O(n)$;
  \item otherwise, the proportion of 0s among nodes of $V$ stays at least $\alpha$ forever.
  \end{itemize}
\end{lemma}
\begin{proof}
  Let ${\alpha = 4/7(1-\epsilon)}$, for $\epsilon> 0$ small enough so that $\alpha>1/2$, and consider the graph of Figure~\ref{fig:amplifiergadget} made of $v$ connected to a triangle and $n$ copies of an identical strip of 14 nodes.
  Moreover, consider the initial configuration $x$ being the one appearing in the figure with $v$ in state $0$.
  Note that $x$ contains a proportion of nodes in state $0$ at least ${\frac{4}{7}\times\frac{14n}{14n+4}}$, \textit{i.e.} at least $\alpha$ if $n$ is supposed large enough.
  Note also that configuration $x$ is stable (the state of each node corresponds to the local majority seen at this node), except possibly for node $v$, depending of the state of its two neighbors outside the gadget.
  Precisely, if the two neighbors of $v$ outside the gadget are in state $1$ then the local majority seen at $v$ changes and $v$ turns into state $1$, otherwise $v$ stays unchanged.
  If at some time step node $v$ turns into $1$, then the two central nodes of the first strip on its right will turn into $1$ the next step.
  From that point on, they will never turn back to $0$ whatever the behavior at $v$ since the sole strip gives them a majority of neighbors in state $1$.
  Moreover the central nodes of the next strips will also progressively turn into state $1$ so that after $O(n)$ steps the gadget reaches a stable configuration identical to $x$ except that the pair of central nodes of each strip is in state $1$, as well as $v$ (which is forced to $1$ by the first strip).
  In this configuration, the proportion of nodes in state $1$ is at least $\alpha$ (for $n$ large enough as above).

  In summary, if the two neighbors of $v$ outside the gadget are both in state $1$ at some step, then, after a linear time, the proportion of $0$ in the gadget stabilizes to at most ${1-\alpha}$.
  Otherwise, the gadget maintains a proportion of state $0$ at least $\alpha$ forever.
\end{proof}

\begin{proof}[Proof of Theorem~\ref{theo:lowerbound}]
For any $n$, any input $a\in\{0,1\}^n$  for Alice and any input $b\in\{0,1\}^n$ for Bob, consider a sequencer gadget $G_A$ for sequence $a$ of size $n_A$, with input $x_A\in\{0,1\}^{V_A}$ and distinguished node $v_A$, and similarly a sequencer gadget $G_B$ for sequence $b$ (by Lemma~\ref{lem:seqgadget}). Since sequencer gadgets are $\cO(n)$, we can choose $m\in \Theta(n)$ such that $\alpha\cdot m > (1-\alpha)m + n_A+n_B$ where $\alpha$ is the constant from the amplifier gadget construction (Lemma~\ref{lem:ampligadget}).  Then consider an amplifier gadget $G$ of parameter $m$ with distinguished node $v$ and initial configuration $x$. Finally let $H$ be the graph made of the disjoint union of $G_A$, $G_B$ and $G$ and where node $v$ is connected to both $v_A$ and $v_B$. Consider the initial configuration $y$ which is equal to $x_A$, $x_B$ and $x$ on $V_A$, $V_B$ and $V$ respectively. $H$ is by construction of size $\Theta(n)$.

Observe that the connections of gadgets inside $H$ and the choice of initial configuration $y$ satisfy the hypothesis of Lemma~\ref{lem:seqgadget} and \ref{lem:ampligadget}. Therefore, by Lemma~\ref{lem:seqgadget}, the sequence of states taken by node $v_A$ will be exactly $a_1,a_2,\ldots,a_n,a_n,\ldots$ and similarly for node $v_B$ with sequence $b$. Then, by Lemma~\ref{lem:ampligadget}, the proportion of 0s inside $G$ will converge towards at most $(1-\alpha)$ if there is $i$ such that $a_i=b_i$ and towards at least $\alpha$ else. By choice of $m$, the steady limit majority in the entire graph $H$ will be for state $1$ in the first case, and for state $0$ in the second case. Moreover, there is convergence to a fixed point in linear time in both cases.

To conclude, observe that a locally checkable proof with certificate $o(n)$ applied to $H$ with the proper initial configuration on each gadget, and some $T\in O(n)$ would give a protocol with $o(n)$ communication to solve \textsf{DISJ} problem on input $(a,b)\in\{0,1\}^n\times\{0,1\}^n$ since $G_A$ is only connected to its complement in $H$ by one edge.

\end{proof}

\subsection{Lower bounds for bounded degree graphs}
\newcommand{\predt}{\textsc{LOC-PRED}}
\newcommand{\EQ}{\textsc{EQ}}
In this section, we study the lower bounds for $\electionpred$ on bounded degree graphs.

\newcommand{\Path}{\textsc{Cycle}}

In particular, we study the case in which $\mathcal{G}_2$ is the class of graphs with maximum degree at most $2.$ In other words, we focus in studying path graphs and cycle graphs. In this case, we show that $\electionpred$ admits proof-label schemes of size $\Omega(\log n).$ We accomplish this task by a reduction to the task of verifying if $G$ is a path or a cycle. More precisely, we define the problem $\Path = \{G \in \mathcal{G}_2: G \text{ is a path graph.} \}$

We recall the notation $P_n$ for a path and $C_n$ for a cycle  of $n$-nodes. 
\begin{proposition} \label{prop:att1}
     Let $n \geq 1,$ $G \in \mathcal{G}_2$ with $n$ nodes and $k \in [n].$ Let us consider the configuration $ x \in \{0,1\}^{n}$ for majority in $G$ given by $x_k = x_{k+1} = 1$ and $x_{k+j \mod{n}} = (1 + x_j) \mod{2}$ for $2 \leq j < n.$ For a node $i \in [n]$ we have that $i$ stabilizes in time $T=\mathcal{O}(n).$ Moreover, for all $i \in [n]$ $x^{t}_i = 1$ for all $t\geq n.$ In addition, node $k$ and $k+1$ do not change their state, i.e. $x^{t}_{k} =x^{t}_{k+1} = 1 $ for all $t\geq0$ and for any $i \in [n],$ such that $x_i=0,$ we have that $i$ changes at least time, more precisely, $x^{n}_i = 1$.
\end{proposition} 

\begin{proof}
    Without loss of generality, let us assume $k=1.$ For $n=4$ we have that $x=1101.$ By the definition of majority rule, it is clear that $F(x) = 1111$ and $x^2(x) = 1111.$ Thus, $F^{t}(x) = 1111$ for $t\geq 4.$ Let us assume that this is true for $n.$ We claim that the proposition holds for $n+1.$ In fact, we have by induction hypothesis that $x^{n}(x)_i = 1 $ for each $i \in [n],$ $F^{n}(x)_n = 1$ and $x^{n}_{n+1}=0$ By the majority rule, since $n+1$ has two neighbors in state $1$ we deduce that $x^{n+1}_{n+1}=1.$ The proposition holds.  
    
    \end{proof}

\begin{proposition} \label{prop:att2}
     Let $n \geq 1$ and let as consider the configuration $y \in \{0,1\}^{2n}$ in $C_{2n}$ given by $
    y_1 \in \{0,1\}$ and $y_{j+1 \mod{n}} = (1 + y_j) \mod{2}$ for $1 \leq j < 2n.$ For any node $i \in [2n]$ we have that $i$ stabilizes in time $T=0.$ Moreover, for all $i \in [n]$ $y^{2t}_i = y_i$ and $y^{2t+1} = 1+y_i \mod{2} $ for all $t\geq 0.$  
\end{proposition} 

\begin{proof}
    The proposition holds since each node in $i$ $C_{2n}$ satisfies $x_{i} = 1+x_{i},$ thus, each node has always two neighbors in the complementary state.
\end{proof}

\begin{proposition} \label{prop:att3}
     Let $n \geq 1.$ If $n$ is odd, the majority dynamics on $P_{n}$ and $C_{n}$ exhibit only fixed points. If $n$ is even then, the configuration $y$ in Proposition \ref{prop:att2} is the only attractor that is not a fixed point. Moreover, $y$ is not reachable.
\end{proposition} 

\begin{proof}
    The proposition follows from the fact that if a node $v$ has a neighbor in the same state, by the majority rule, it cannot change its state.
\end{proof}

\begin{lemma}
    Let us suppose that $\electionpred$ restricted to $\mathcal{G}_2$ admits a locally checkable proof with certificates of size $L$. Then, $\Path$ admits a locally checkable proof with certificates of size $2L.$
\end{lemma}
\begin{proof}
 Let $G$ be an instance of $\Path$ of size $n.$ Let us fix $T=2n.$ Let $\pi_{P} = (\mathcal{M}_p,\mathcal{D}_p)$ be a PLS for $\electionpred(P_{2n},x, \Vec{1},T)$ of size $L$.   
 We define the marking $\mathcal{M}'(v) = (M_p(v),M_c(v),y_v),$ where $y$ is defined as in Proposition \ref{prop:att2}. We describe the decoding algorithm $\mathcal{D'}$ for each node $v$:

 \begin{enumerate}
     \item \textbf{If} $\delta(v) = 1$ or $\delta(v) = 2$, node $v$ creates neighbors $v_1,v_2.$ 
     \item \textbf{Else} node $v$ \textit{rejects.}
     \item \textbf{Verification:}
     \item Node $v$ verifies that $y_w = 1+ y_v \mod{2}$ for each $w \in N(v)$. Otherwise \textit{rejects.}
     \item \textbf{If} $\delta(v) = 1$, node $v$ assigns $y_{v_2} = y_{v_1}=1$ and \textit{accept}.
     \item \textbf{If} $\delta(v) = 2$, node $v$ assigns $y_{v_2} = 1+ y_{v_1} \mod{2}$ and \textit{rejects}.
     \item Let $G'$ be the new graph obtained in the previous steps. 
     Observe that $G'$ is whether $P_{2n}$ or $C_{2n}.$ 
     In addition, we have that each node $v$ can test all the possible labeling for its neighbors. 
     This latter observation together with the fact that each node has the certificates for $P_{2n}$ imply that each node $v$ in $G$ can run $\pi_P$  for instance $\electionpred(G',x, \Vec{1},T)$. 
     Then, node $v$ accepts if and only if it accepts on $\pi_P.$
     \end{enumerate}
We claim that $\pi'$ is a PLS for $\Path.$ 

\emph{Completenesss.} Observe that $\Path(G) = 1$ if and only if there is at least one node $v$ such that $\delta(v)=1.$ On one hand, by definition, the configuration $x$ is the one in Proposition \ref{prop:att1}. Then,  the attractor will be $1,$ and thus $\pi'$ accepts. 

\emph{Soundness.} On the other hand, if  $\Path(G)=0$ then, for all labeling $\mathcal{L},$ by Proposition \ref{prop:att2} $\pi_P$ must reject by Proposition \ref{prop:att3}. We conclude that $\pi'$ is a PLS for $\Path$ of size $2L+1.$
\end{proof}

In \cite{goos2016locally} it is shown that every locally checkable proof for $\Path$ has certificates of size $\Omega(\log n)$. We obtain the following theorem.

\begin{theorem}
    Every locally checkable proof for $\electionpred$ on \(n\)-node graphs of maximum degree \(2\) has certificates on $\Omega(\log n)$ bits.
\end{theorem}
\begin{proof}
    Assume that there exists a locally checkable proof $\pi$ for  $\electionpred$ of size $o(\log n ).$ Thus, by the previous lemma, there must be a locally checkable proof $\pi'$ for $\Path$ of size $o(\log n).$ This is a contradiction with the fact that the problem $\Path$ admits a locally checkable proof with size $\Omega(\log n)$. The theorem holds.
\end{proof}

%
%
%
\bibliographystyle{splncs04} 
\bibliography{biblio}   

\end{document}